\documentclass{article}
\usepackage{amsmath}
\usepackage{amssymb}
\usepackage{amsthm}
\usepackage{authblk}
\usepackage{graphicx}
\usepackage{natbib}
\usepackage{xcolor}
\usepackage[hidelinks]{hyperref}
\usepackage[a4paper,top=3cm,bottom=3.5cm,left=2.5cm,right=2.5cm]{geometry}

\title{Towards Optimal Prefix-Free Graph Construction: NP-Hardness and Structural Insights}
\author[1]{Andrej Baláž}
\author[2,3]{Alexandru Popa}
\affil[1]{Institute of Clinical and Translational Research, Biomedical Research Center of the Slovak Academy of Sciences}
\affil[2]{Faculty of Mathematics and Computer science University of Bucharest, Romania}
\affil[3]{National Institute for Research and Development in Informatics – ICI Bucharest}

\newtheorem{definition}{Definition}
\newtheorem{problem}{Problem}
\newtheorem{theorem}{Theorem}
\newtheorem{lemma}{Lemma}
\newtheorem{claim}{Claim}
\newtheorem{observation}{Observation}
\newtheorem{corollary}{Corollary}
\newtheorem{example}{Example}

\begin{document}
\maketitle

\begin{abstract}
Prefix-free parsing provides an efficient way to construct compressed representations of large and repetitive pangenomes and naturally induces a graph representation known as a prefix-free graph. In this work, we initiate a theoretical study of the problem of constructing prefix-free graphs of minimum size, where the size accounts for both the total length of distinct segment labels and the paths representing the input sequences.

We show that selecting an optimal set of trigger words is NP-hard, already when triggers consist of single characters. Using a synchronized-code reduction, we extend this hardness result to every fixed trigger length and further show that the problem remains NP-hard over an alphabet of size three. We then establish a structural connection between prefix-free graphs and de Bruijn graphs. In particular, we show that every compacted de Bruijn graph can be realized as a prefix-free graph and derive a hierarchy relating the sizes of minimum pangenomic graphs, minimum prefix-free graphs, compacted de Bruijn graphs, and de Bruijn graphs. Finally, we give an exact fixed-parameter algorithm running in $O(2^q n)$ time, where $q$ is the number of distinct candidate trigger words and $n$ is the total pangenome length.

Our results characterize both the computational limitations and the structural properties of optimizing prefix-free graph representations and provide a theoretical foundation for the design of compact graph representations of repetitive pangenomic data.
\end{abstract}

\section{Introduction}
Sequence alignment is one of the fundamental operations in bioinformatics. In a typical read-mapping workflow, sequencing reads are aligned to a reference genome so that differences between the reads and the reference can be identified and, ultimately, associated with phenotypes of interest. Traditionally, the reference is represented by a single finite string over the alphabet $\Sigma=\{A,C,G,T\}$. A single linear reference, however, can induce \emph{reference bias}: inferred sequences may be systematically pulled toward the chosen reference rather than reflecting the variation supported by the data~\citep{computational2018computational}. One way to mitigate this bias is to represent and analyze multiple related genomes jointly, that is, to use a pangenome.

The term \emph{pangenome} was introduced by Tettelin et al.~\citep{tettelin2005pangenome} in their study of variation in \emph{Streptococcus agalactiae}. The Computational Pan-Genomics Consortium defines a pangenome broadly as a set of genomic sequences intended to be analyzed jointly~\citep{computational2018computational}. In practice, pangenomic collections often contain genomes from closely related organisms and are therefore both large and highly repetitive. These two properties make the choice of representation algorithmically important.

A common approach is to represent a pangenome as a graph. Similar genomic regions are shared, while variation is represented by alternative walks through the graph. Variation graphs~\citep{garrison2018graphical}, for example, can encode a wide range of variation, from single-nucleotide substitutions to large structural rearrangements. Their flexibility also means that the same set of sequences can admit many non-isomorphic graph representations, with no canonical choice of the most useful one. Moreover, Equi et al.~\citep{equi2023graphs} showed, assuming the Strong Exponential Time Hypothesis, that general labeled graphs do not admit polynomial-size indexes supporting truly subquadratic pattern matching. This has motivated the study of more structured graph classes, including cactus graphs~\citep{2011paten1,2011paten2}, Wheeler graphs~\citep{gagie2017wheeler}, elastic-degenerate strings~\citep{bernardini2017pattern,iliopoulos2017efficient}, founder block graphs~\citep{makinen2020linear}, and $p$-sortable graphs~\citep{cotumaccio2021indexing}. Many such constructions, however, rely on pairwise or multiple-sequence alignment and can therefore be expensive on large pangenomes.

A complementary approach is to keep the pangenome as a collection of strings and exploit repetitiveness inside compressed string-processing algorithms. This approach builds on a mature toolkit that includes the Burrows--Wheeler transform, suffix arrays, suffix trees, and compressed variants of these structures. The central practical difficulty is to construct such data structures for very large repetitive collections using close-to-linear time and substantially less working space than the uncompressed input.

Prefix-free parsing (PFP), inspired by context-triggered piecewise hashing~\citep{kornblum2006identifying}, addresses this difficulty by partitioning strings at positions selected by a rolling hash. The partition produces a dictionary of distinct phrases together with a parse over phrase identifiers. This representation has been used to construct a variety of compressed data structures efficiently, including the Burrows--Wheeler transform, compressed suffix trees, Wheeler graphs, FM-indexes, and Lempel--Ziv factorizations~\citep{2019rpair,2019bigBWT,2021pfp_cst,2022pfp_wg,2023pfg,2023recursive_bigBWT,2023pfp_lz77,2023pfp_fm_index,2024recursive_pfp_pangenome,2024recursive_repair,luca2025measuring}.

In previous work~\citep{2023pfg}, we observed that the dictionary and parse produced by prefix-free parsing naturally induce a graph, which we called a \emph{prefix-free graph} (PFG). Prefix-free graphs can be constructed much faster than alignment-based pangenome graphs, support the construction of several compressed string indexes, and admit a precise combinatorial description that makes their structure amenable to theoretical analysis. They are also closely related to de Bruijn graphs: choosing all trigger words of a fixed length yields the corresponding node-centric de Bruijn representation, while removing triggers merges consecutive non-branching regions.

\subsection{Our results}
In this work, we initiate a theoretical study of the size of prefix-free graphs and of the computational complexity of constructing an optimal one. We measure the size of a pangenomic graph as the total length of its distinct segment labels plus the total number of vertices in the paths representing the input sequences. This objective captures a natural trade-off: longer segments shorten the paths, whereas shorter segments can expose more repeated substrings and thereby reduce the total dictionary size.

We formalize optimization problems for unrestricted pangenomic graphs, pangenomic graphs with consistent segmentation, prefix-free graphs, and prefix-free graphs whose trigger words must be selected from a prescribed candidate set. Our first main result is computational: selecting trigger words that minimize the size of the resulting prefix-free graph is NP-hard. We first prove NP-hardness when trigger words have length one, by a reduction from \textsc{1-in-3-SAT}. We then give a synchronized-code reduction showing that the problem remains NP-hard for every fixed trigger length $k\geq 1$. A second instantiation of the same reduction shows NP-hardness even when the target alphabet is restricted to three symbols. The restricted-trigger variant is therefore NP-hard as well.

Our second contribution concerns the relationship between prefix-free graphs and de Bruijn graphs. For a de Bruijn graph of order $k$, compaction preserves the total length of the \emph{segment bases}, obtained by removing the final $k-1$ symbols from each segment. We show that this quantity equals the number of distinct $k$-mers and lower-bounds the segment-base contribution of every prefix-free graph whose trigger words have length $k-1$. Conversely, the natural trigger set consisting of branching $(k-1)$-overlaps realizes the compacted de Bruijn graph as a prefix-free graph. Consequently, for every order $k$,
\[
    2\sqrt{n}
    \leq |MPG|
    \leq |MPFG|
    \leq |cDBG_k|
    \leq |DBG_k|,
\]
where $n$ is the total length of the pangenome. In addition, choosing no trigger words gives the general bound $|MPFG|\leq n+r$, where $r$ is the number of input strings, and $|DBG_k|\leq (k+1)n$. For a single input string, the former simplifies to $|MPFG|\leq n+1$.

Finally, if $q$ distinct $k$-mers can occur as trigger words, exhaustive enumeration yields an exact fixed-parameter algorithm running in $O(2^q n)$ time. Thus the optimization problem is intractable in general but remains tractable when the number of distinct candidate triggers is small, a regime that is particularly relevant for highly repetitive pangenomic data.

\section{Preliminaries and problem definition}
\label{sec:preliminaries}

A \emph{pangenome} is a finite collection of strings over an alphabet $\Sigma$. We allow repeated input strings; repeated occurrences contribute separately to path lengths, while segment labels are stored only once.

\begin{definition}[Pangenomic graph]
\label{def:pangenomic_graph}
A pangenomic graph is a directed graph $G=(V,E,P)$, where $V$ is a set of vertices, $E$ is a set of directed arcs, and $P$ is a collection of directed paths indexed by the input-string occurrences. Each vertex $v\in V$ is labeled by a string $\ell(v)\in\Sigma^*$, called a \emph{segment}. An arc $(u,v)\in E$ indicates that the segment $\ell(v)$ follows the segment $\ell(u)$ in at least one represented input string. For every input-string occurrence there is a path $p=(v_1,\ldots,v_t)\in P$ whose segment labels reconstruct that string under the representation-specific joining operation.
\end{definition}

Distinct vertices of a general pangenomic graph may carry the same segment label. For a path $p$, let $|p|$ denote its number of vertices.

\begin{definition}[Size of a pangenomic graph]
\label{def:size_pangenomic}
For a pangenomic graph $G=(V,E,P)$, define
\[
    |G|=\sum_{v\in V}|\ell(v)|+\sum_{p\in P}|p|.
\]
\end{definition}

A \emph{sentinel} is a character not contained in $\Sigma$; we write $\$$ for a sentinel. A string ending in $\$^k$ is called \emph{$k$-sentinelled}. Sentinels are useful in implementations to force a final parsing boundary, but the reductions below use the equivalent convention that trigger occurrences at the first or last possible position are ignored.

A \emph{$k$-mer} is a string of length $k$. A \emph{trigger word} is a $k$-mer selected as a parsing boundary. We first define the joining operation used to reconstruct a string from overlapping segments. If $x=x't$ and $y=ty'$ for the same $k$-mer $t$, define
\[
    x\diamond y=x'ty'.
\]
Thus $x\diamond y$ concatenates the two strings while retaining only one copy of their common trigger word.

\begin{definition}[Context-triggered partition]
\label{def:seg}
Let $s=s[1]\cdots s[m]$ be a string and let $W_k\subseteq\Sigma^k$ be a set of trigger words. Let
\[
    1<p_1<p_2<\cdots<p_t\leq m-k
\]
be exactly the starting positions of the trigger occurrences that are strictly internal to $s$, i.e., $s[p_j\ldots p_j+k-1]\in W_k$. If $t=0$, the partition consists only of $s$. Otherwise define
\begin{align*}
    s_1 &= s[1\ldots p_1+k-1],\\
    s_j &= s[p_{j-1}\ldots p_j+k-1] \qquad (2\leq j\leq t),\\
    s_{t+1} &= s[p_t\ldots m].
\end{align*}
Then
\[
    s=s_1\diamond s_2\diamond\cdots\diamond s_{t+1}.
\]
The context-triggered partition of a pangenome is the set of all \emph{distinct} segment strings produced in this way; it is not a multiset.
\end{definition}

\begin{definition}[Prefix-free graph]
\label{def:prefix-free}
Given a pangenome and a trigger set $W_k\subseteq\Sigma^k$, let $\mathcal S$ be the set of distinct segments produced by Definition~\ref{def:seg}. The induced prefix-free graph $G=(V,E,P)$ is defined as follows.
\begin{enumerate}
    \item There is exactly one vertex $v_s\in V$ for every segment $s\in\mathcal S$, labeled by $s$.
    \item If an input string has partition $s_1\diamond\cdots\diamond s_t$, then $(v_{s_i},v_{s_{i+1}})\in E$ for every $1\leq i<t$.
    \item The same input-string occurrence contributes the path $(v_{s_1},\ldots,v_{s_t})$ to $P$.
\end{enumerate}
\end{definition}

The size of a prefix-free graph is measured by Definition~\ref{def:size_pangenomic}.

We now define the optimization problems studied in this paper.

\begin{problem}[Minimum pangenomic graph (MPG)]
\label{prob:mpgc}
Given a pangenome, find a pangenomic graph of minimum size.
\end{problem}

\begin{problem}[Minimum pangenomic graph with consistent segmentation (MPGcs)]
\label{prob:mpgc-restricted}
Given a pangenome, find a minimum-size pangenomic graph subject to the following consistency condition: if a segment $x$ occurs in an input string, then every occurrence of $x$ in every input string must coincide with a segment occurrence in the corresponding path representation.
\end{problem}

\begin{example}
Consider the single sequence \texttt{TAGCTTAGTCGGATGCTCGATAG}. The segmentation
\[
\texttt{TAGCTTAG}\;\texttt{TCG}\;\texttt{GATGC}\;\texttt{TCG}\;\texttt{ATAG}
\]
is consistent. In contrast,
\[
\texttt{TAGCTTAG}\;\texttt{TCG}\;\texttt{GAT}\;\texttt{GCTCG}\;\texttt{ATAG}
\]
is not, because \texttt{TCG} is itself a segment but also occurs strictly inside the segment \texttt{GCTCG}.
\end{example}

\begin{problem}[Minimum prefix-free graph ($k$-MPFG)]
\label{prob:mpfgc}
Given a pangenome and a trigger length $k$, find a trigger set $W_k\subseteq\Sigma^k$ minimizing the size of the induced prefix-free graph.
\end{problem}

\begin{problem}[Minimum prefix-free graph with restricted trigger set ($k$-MPFGrt)]
\label{prob:mpfgc-restricted}
Given a pangenome, a trigger length $k$, and a candidate set $C_k\subseteq\Sigma^k$, find a subset $W_k\subseteq C_k$ minimizing the size of the induced prefix-free graph.
\end{problem}

\section{NP-hardness}
\label{sec:hardness}

We prove NP-hardness of $k$-MPFG in three steps. We first handle $k=1$, then lift the construction to every fixed $k$, and finally encode the alphabet over three symbols. It is enough to consider the decision version asking whether a trigger set of cost at most a given threshold exists.

\subsection{Trigger words of length one}

We reduce from monotone \textsc{1-in-3-SAT}~\citep{1983_theory_of_np}. The input is a 3-CNF formula in which every literal is non-negated, and the question is whether there exists a truth assignment making exactly one variable in every clause true. By duplicating clauses if necessary, we may assume without loss of generality that the formula has $m\geq 5$ clauses.

Let the variables be $x_1,\ldots,x_n$ and the clauses be $C_1,\ldots,C_m$. The alphabet contains one symbol $v_i$ for each variable and three fresh symbols $d_i^1,d_i^2,d_i^3$ for each clause. Set
\[
    \ell_1=\left\lceil\frac{m^2}{5}\right\rceil,
    \qquad
    \ell_2=m^2,
    \qquad
    D_i^j=(d_i^j)^{\ell_1}.
\]

Fix a clause $C_i=x_{i_1}\vee x_{i_2}\vee x_{i_3}$. For each $j\in\{1,2,3\}$ define a cyclic ordering
\[
(a_{i,j,1},a_{i,j,2},a_{i,j,3})=
\begin{cases}
(v_{i_1},v_{i_2},v_{i_3}), & j=1,\\
(v_{i_2},v_{i_3},v_{i_1}), & j=2,\\
(v_{i_3},v_{i_1},v_{i_2}), & j=3.
\end{cases}
\]
For every $j$, add $\ell_2$ copies of
\[
    r_i^j=D_i^j a_{i,j,1}D_i^j a_{i,j,2}D_i^j a_{i,j,3}D_i^j
\]
and the following six strings:
\begin{align*}
SP_i^j &= D_i^j a_{i,j,1}, &
SS_i^j &= a_{i,j,1}D_i^j a_{i,j,2}D_i^j a_{i,j,3}D_i^j,\\
MP_i^j &= D_i^j a_{i,j,1}D_i^j a_{i,j,2}, &
MS_i^j &= a_{i,j,2}D_i^j a_{i,j,3}D_i^j,\\
EP_i^j &= D_i^j a_{i,j,1}D_i^j a_{i,j,2}D_i^j a_{i,j,3}, &
ES_i^j &= a_{i,j,3}D_i^j.
\end{align*}
Under any trigger set that avoids the delimiter symbols, every segment produced inside group $(i,j)$ contains $d_i^j$. Since the delimiter symbols are unique to their groups, segment labels from different groups cannot coincide, and the group costs therefore add.

For one group, direct evaluation of Definition~\ref{def:seg} gives the costs in Table~\ref{tab:gadget-costs}. Here $a_1,a_2,a_3$ denote the three variables in that group's cyclic order.

\begin{table}[h]
\centering
\begin{tabular}{c|c}
selected variable triggers & group cost \\
\hline
$\varnothing$ & $16\ell_1+\ell_2+21$ \\
$\{a_1\}$ or $\{a_3\}$ & $10\ell_1+2\ell_2+20$ \\
$\{a_2\}$ & $8\ell_1+2\ell_2+18$ \\
$\{a_1,a_2\}$ or $\{a_2,a_3\}$ & $6\ell_1+3\ell_2+18$ \\
$\{a_1,a_3\}$ & $6\ell_1+3\ell_2+19$ \\
$\{a_1,a_2,a_3\}$ & $4\ell_1+4\ell_2+18$
\end{tabular}
\caption{Cost of a clause group for each subset of its three variable symbols. The two omitted cases follow by left--right symmetry.}
\label{tab:gadget-costs}
\end{table}

As a representative calculation, if only the middle symbol $a_2$ is a trigger, each copy of $r_i^j$ is split into two segments, contributing $2\ell_2$ path vertices. Across $r_i^j$ and the six auxiliary strings, the distinct segment labels have total length $8\ell_1+10$, while the six auxiliary paths contribute $8$ vertices. Hence the total is $8\ell_1+2\ell_2+18$, as in the third row of Table~\ref{tab:gadget-costs}.

Selecting a delimiter symbol $d_i^j$ is never optimal at the threshold used below. Indeed, the $\ell_2$ copies of $r_i^j$ alone then have path length at least $\ell_2(4\ell_1-1)$. Define
\[
    Q=6\ell_2+28\ell_1+58.
\]
For $m\geq5$,
\[
\ell_2(4\ell_1-1)-Q
=4\ell_1\ell_2-7\ell_2-28\ell_1-58>0,
\]
using $\ell_2=m^2$ and $m^2/5\leq\ell_1\leq m^2/5+1$. Thus a delimiter trigger already forces a single clause group above the target contribution $Q$.

\begin{theorem}
\label{thm:np-hard-k1}
$1$-MPFG is NP-hard.
\end{theorem}

\begin{proof}
We show that the monotone \textsc{1-in-3-SAT} instance is satisfiable if and only if the constructed pangenome admits a trigger set whose graph has size at most $mQ$.

Suppose first that the formula has a satisfying assignment with exactly one true variable in each clause. Select $v_i$ as a trigger exactly when $x_i$ is true. In every clause, the unique selected variable appears once in the first position of a cyclic group, once in the middle position, and once in the last position. Therefore the three groups of the clause contribute
\begin{align*}
&2(10\ell_1+2\ell_2+20)+(8\ell_1+2\ell_2+18)\\
&\qquad=6\ell_2+28\ell_1+58=Q.
\end{align*}
Summing over the $m$ clauses gives total size $mQ$.

Conversely, consider any trigger set of total size at most $mQ$. No delimiter symbol $d_i^j$ can be selected, by the preceding bound. We therefore interpret $x_i$ as true exactly when $v_i$ is selected. It remains to show that every clause has exactly one selected variable.

If no variable of a clause is selected, its three groups contribute
\[
    3(16\ell_1+\ell_2+21),
\]
which exceeds $Q$ because
\[
    3(16\ell_1+\ell_2+21)-Q
    =20\ell_1-3\ell_2+5>0.
\]
If exactly two variables are selected, the cyclic order produces two adjacent-pair cases and one endpoint-pair case, for a total of
\[
    18\ell_1+9\ell_2+55,
\]
which exceeds $Q$ because
\[
    (18\ell_1+9\ell_2+55)-Q
    =3\ell_2-10\ell_1-3>0
\]
for $m\geq5$. If all three variables are selected, the contribution is
\[
    12\ell_1+12\ell_2+54,
\]
and
\[
    (12\ell_1+12\ell_2+54)-Q
    =6\ell_2-16\ell_1-4>0.
\]
Thus every clause contributes at least $Q$, with equality only when exactly one of its variables is selected. Since the total cost is at most $mQ$, all clauses must attain equality. The selected variables therefore define a satisfying \textsc{1-in-3-SAT} assignment.
\end{proof}

\subsection{A synchronized-code reduction}

The next lemma lifts the preceding hardness result while preserving the objective up to an affine transformation.

\begin{lemma}[Synchronized-code reduction]
\label{lem:synchronized-code}
Let $\mathcal S$ be an instance of $1$-MPFG over alphabet $\Sigma$, and let $\varphi:\Sigma\to\Gamma^k$ be an injective fixed-length code with the synchronization property that, in every concatenation $\varphi(a)\varphi(b)$, a codeword $\varphi(c)$ occurs only at one of the two aligned codeword positions. Suppose moreover that there is a guard word $g\in\Gamma^k$ such that no codeword occurs in $g t g$ for any non-codeword $k$-mer $t$ occurring in a concatenation $\varphi(a)\varphi(b)$, and no string $g t g$ occurs in an encoded input string. Then one can construct in polynomial time a $k$-MPFG instance $\mathcal S'$ and a constant $C$ such that
\[
    \operatorname{OPT}(\mathcal S')=k\operatorname{OPT}(\mathcal S)+C.
\]
\end{lemma}

\begin{proof}
Let $N$ be the total length of the strings in $\mathcal S$. For $s=a_1\cdots a_m$, write
\[
    \widehat{s}=\varphi(a_1)\cdots\varphi(a_m).
\]
The first part of $\mathcal S'$ contains $k$ copies of $\widehat{s}$ for every input-string occurrence $s\in\mathcal S$.

Let $R$ be the set of non-codeword $k$-mers that occur in some $\varphi(a)\varphi(b)$. Every non-codeword trigger that can affect an encoded input string belongs to $R$. For each $t\in R$, define the penalty gadget
\[
    G_t=g t g.
\]
Let $L$ be the total length of one copy of every distinct encoded source string together with one copy of every gadget $G_t$, and set
\[
    B=(k+1)L,
    \qquad
    M=B+1.
\]
Add $M$ copies of every $G_t$ to $\mathcal S'$. The construction is polynomial because $|R|\leq |\Sigma|^2(k-1)$.

Consider a source trigger set $W_1\subseteq\Sigma$ and the target trigger set
\[
    \widehat W=\{\varphi(a):a\in W_1\}.
\]
By synchronization, the partition of every encoded source string is exactly the encoded partition of the original string. If the source graph has segment-length contribution $X$ and path-length contribution $Y$, the encoded source strings contribute $kX+kY=k(X+Y)$: segment lengths are multiplied by $k$, and the $k$ copies multiply the path contribution by $k$.

By the guard assumption, no codeword trigger occurs in any $G_t$. Hence every penalty gadget remains a single segment and contributes a constant independent of $W_1$. The total target cost is therefore $k|G(W_1)|+C$ for a computable constant $C$.

It remains to exclude non-codeword triggers from an optimum. Suppose an optimal target trigger set contains $t\in R$. The central occurrence of $t$ is strictly internal in each of the $M$ copies of $G_t$. Removing $t$ therefore decreases the total path length by at least $M$. On the other hand, for any string $x$ of length $\ell$, a context-triggered partition has total segment length along its path at most $(k+1)\ell$: each internal boundary introduces an overlap of exactly $k$ symbols, and there are at most $\ell$ such boundaries. Consequently, the total length of the global dictionary is at most $B$. Removing one trigger can therefore increase the dictionary contribution by at most $B$. Since $M=B+1$, removing $t$ strictly decreases the graph size, contradicting optimality.

Thus every optimal target trigger is a codeword. Mapping codeword triggers back through $\varphi$ gives a source trigger set, and the affine cost relation follows in both directions.
\end{proof}

\begin{theorem}
\label{thm:np-hard-genk}
For every fixed $k\geq1$, $k$-MPFG is NP-hard.
\end{theorem}

\begin{proof}
The case $k=1$ is Theorem~\ref{thm:np-hard-k1}. Fix $k\geq2$. For every $a\in\Sigma$, introduce a fresh symbol $a'$ and define
\[
    \varphi(a)=a(a')^{k-1}.
\]
A codeword begins with an unprimed symbol and all remaining positions contain the corresponding primed symbol. Hence a codeword can occur in a concatenation of codewords only at an aligned position. Add one further fresh guard symbol $\#$ and choose $g=\#^k$. The hypotheses of Lemma~\ref{lem:synchronized-code} are immediate, so $1$-MPFG reduces to $k$-MPFG.
\end{proof}

\begin{theorem}
\label{thm:np-hard-alpha3}
MPFG is NP-hard even when the input alphabet has size three.
\end{theorem}

\begin{proof}
Reduce from the $1$-MPFG instances of Theorem~\ref{thm:np-hard-k1}. Let $\beta=\max\{1,\lceil\log_2|\Sigma|\rceil\}$ and fix an injective binary encoding $e:\Sigma\to\{0,1\}^{\beta}$. Set
\[
    k=\beta+2,
    \qquad
    \varphi(a)=2e(a)2
\]
over the alphabet $\Gamma=\{0,1,2\}$.

This code is synchronized. A codeword starts with $2$, is followed by a binary symbol, and ends with $2$. In a concatenation $\varphi(a)\varphi(b)$, an unaligned length-$k$ window either starts with $0$ or $1$, or starts at the final $2$ of $\varphi(a)$ and therefore has $2$ as its second symbol. In neither case can it equal a codeword.

Choose the guard word $g=0^k$. Since every codeword begins and ends with $2$, no codeword occurs in $g$. If $t$ is a non-codeword $k$-mer arising from two consecutive codewords, then no codeword can occur in $g t g$: any such occurrence would have to be entirely contained in the middle block $t$, forcing $t$ itself to be a codeword. Moreover, encoded input strings contain a symbol $2$ at least once every $k$ positions, so $0^k$ cannot occur in them; hence no penalty gadget $g t g$ can coincide with an encoded source segment. Lemma~\ref{lem:synchronized-code} now gives a polynomial reduction to an alphabet of size three.
\end{proof}

\begin{corollary}
\label{cor:restricted-hard}
For every fixed $k\geq1$, $k$-MPFGrt is NP-hard. It remains NP-hard when the alphabet is restricted to three symbols.
\end{corollary}

\begin{proof}
Given an instance of $k$-MPFG, use as candidate set all $k$-mers that occur in the input pangenome. Trigger words that do not occur can never affect the graph and may be omitted. Thus an algorithm for $k$-MPFGrt would solve $k$-MPFG with the same optimum.
\end{proof}

\section{Relationship between prefix-free and de Bruijn graphs}
\label{sec:dbg}

We next relate prefix-free graphs to node-centric de Bruijn graphs. For background on de Bruijn graphs in sequence analysis, see Compeau and Pevzner~\citep{compeau2015bioinformatics}.

\begin{definition}[de Bruijn graph ($DBG_k$)]
\label{def:deBruijngraph}
Given a pangenome and an integer $k$, the node-centric de Bruijn graph $DBG_k=(V,E)$ has one vertex for every distinct $k$-mer occurring in the pangenome. There is an arc $(u,v)\in E$ if the suffix of length $k-1$ of $u$ equals the prefix of length $k-1$ of $v$.
\end{definition}

If the labels of two adjacent vertices overlap in a common $(k-1)$-mer, we use $\diamond$ for the corresponding overlap-aware concatenation.

\begin{definition}[Compacted de Bruijn graph ($cDBG_k$)]
The compacted de Bruijn graph is obtained by repeatedly merging an arc $(u,v)$ for which $u$ has outdegree one and $v$ has indegree one. The merged vertex is labeled by $\ell(u)\diamond\ell(v)$. Equivalently, the vertices of $cDBG_k$ are the maximal non-branching walks of $DBG_k$, represented by their spelled strings.
\end{definition}

For either $DBG_k$ or $cDBG_k$, the input strings induce paths through the graph. We measure their size by Definition~\ref{def:size_pangenomic}, using the spelled string of a vertex as its segment label.

\begin{definition}[Segment base]
For a segment $s$ of a $DBG_k$, $cDBG_k$, or a prefix-free graph with trigger length $k-1$, its \emph{segment base} is the prefix of $s$ obtained by deleting its final $k-1$ symbols. Its base length is therefore $|s|-k+1$, the number of $k$-mer positions in $s$.
\end{definition}

\begin{claim}
\label{claim_dbg_eq_cdbg}
The sum of the segment-base lengths in $DBG_k$ and $cDBG_k$ is the same and equals the number of distinct $k$-mers in the pangenome.
\end{claim}

\begin{proof}
Every $DBG_k$ vertex is a $k$-mer and therefore has base length one. Hence the total base length is exactly the number of distinct $k$-mers. If two compactable segments have base lengths $a$ and $b$, their overlap-aware join has base length $a+b$. Every compaction step therefore preserves the sum of base lengths, proving the claim.
\end{proof}

\begin{claim}
\label{claim_2_cdbg}
Let $G$ be any prefix-free graph whose trigger words have length $k-1$. The sum of the segment-base lengths of $G$ is at least the sum of the segment-base lengths of $cDBG_k$.
\end{claim}

\begin{proof}
Every $k$-mer occurring in an input string is contained in at least one segment of $G$. Indeed, adjacent PFG segments overlap in exactly $k-1$ symbols, so a $k$-mer crossing a parsing boundary is contained in one of the two adjacent segments. A segment $s$ contains at most $|s|-k+1$ distinct $k$-mers, which is exactly its base length. Therefore the union of the $k$-mer sets of all PFG segments has cardinality at most the sum of their base lengths. Since that union contains every distinct input $k$-mer, the PFG base sum is at least the number of distinct $k$-mers. Claim~\ref{claim_dbg_eq_cdbg} identifies this number with the base sum of $cDBG_k$.
\end{proof}

\begin{claim}
\label{claim:cdbg-is-pfg}
For every $k\geq2$, the compacted de Bruijn graph $cDBG_k$ is realizable as a prefix-free graph with trigger length $k-1$.
\end{claim}

\begin{proof}
For a $(k-1)$-mer $x$, let $L(x)$ be the set of $k$-mers having suffix $x$ and let $R(x)$ be the set of $k$-mers having prefix $x$. An arc whose overlap is $x$ is non-branching exactly when $|L(x)|=|R(x)|=1$. Let $W$ contain precisely the observed $(k-1)$-mers $x$ for which $|L(x)|\neq1$ or $|R(x)|\neq1$.

A context-triggered partition using $W$ cuts an input path exactly at branching overlaps and does not cut at non-branching overlaps. Consequently, each segment spells a maximal non-branching walk of $DBG_k$, and every such walk is represented by one segment. These are exactly the vertices of $cDBG_k$. With the usual terminal sentinel convention, the first and last maximal walks are handled in the same way. Thus the induced PFG and $cDBG_k$ have the same segment labels and input paths.
\end{proof}

\begin{claim}
\label{claim:MPG>sqrtn}
If the total length of the pangenome is $n$, then every pangenomic graph has size at least $2\sqrt n$.
\end{claim}

\begin{proof}
Let
\[
    A=\sum_{v\in V}|\ell(v)|,
    \qquad
    B=\sum_{p\in P}|p|.
\]
For each vertex $v$, let $\operatorname{occ}(v)$ be its total number of occurrences across all paths. Then $B=\sum_v\operatorname{occ}(v)$ and
\[
    n=\sum_v\operatorname{occ}(v)|\ell(v)|
    \leq
    \left(\sum_v\operatorname{occ}(v)\right)
    \left(\sum_v|\ell(v)|\right)
    =AB.
\]
By the arithmetic--geometric mean inequality,
\[
    |G|=A+B\geq2\sqrt{AB}\geq2\sqrt n.
\]
\end{proof}

\begin{theorem}[Size hierarchy]
\label{thm:size-hierarchy}
For every pangenome of total length $n$ and every $k\geq2$,
\[
    2\sqrt n
    \leq |MPG|
    \leq |MPFG|
    \leq |cDBG_k|
    \leq |DBG_k|.
\]
If the pangenome contains $r$ nonempty input strings, then additionally
\[
    |MPFG|\leq n+r,
    \qquad
    |DBG_k|\leq(k+1)n.
\]
\end{theorem}

\begin{proof}
The first inequality is Claim~\ref{claim:MPG>sqrtn}. Every prefix-free graph is a pangenomic graph, so $|MPG|\leq|MPFG|$. Claim~\ref{claim:cdbg-is-pfg} gives $|MPFG|\leq|cDBG_k|$. Every compaction step replaces two overlapping segment labels by their overlap-aware join and weakly shortens every affected input path, so it cannot increase graph size; hence $|cDBG_k|\leq|DBG_k|$.

For the first additional bound, choose the empty trigger set. Each input-string occurrence then contributes one path vertex, while the total length of the distinct full-string segments is at most $n$. Thus $|MPFG|\leq n+r$. Finally, if $q$ is the number of distinct $k$-mers, then the segment contribution of $DBG_k$ is $kq\leq kn$, while the total path length is at most $n$, giving $|DBG_k|\leq(k+1)n$.
\end{proof}

\section{Exact construction for a small trigger universe}

The hardness results concern the general case in which the number of possible trigger words can grow with the input. When the trigger universe is small, an exact algorithm follows immediately.

\begin{observation}
\label{obs:fpt}
Let $q$ be the number of distinct length-$k$ words that occur internally in the pangenome and can therefore affect a context-triggered partition. An optimal prefix-free graph can be found in $O(2^q n)$ time and polynomial space, where $n$ is the total pangenome length.
\end{observation}

\begin{proof}
Enumerate all $2^q$ subsets of the candidate trigger words. For each subset, scan the pangenome once, construct the induced partition, and accumulate the distinct segment lengths and path lengths. With hashing of segment labels, this evaluation takes $O(n)$ expected time per subset. Keeping the best solution gives the stated bound.
\end{proof}

This observation gives a fixed-parameter algorithm parameterized by the number of distinct candidate $k$-mers. In highly repetitive pangenomes, this parameter can be much smaller than the total input length; related repetitiveness measures have been studied in~\citep{2021repmeasure,2021compress}.

\section{Discussion}

The results above separate two aspects of prefix-free graph construction. Once a trigger set is fixed, the graph can be built by a linear scan of the input, whereas selecting an optimal trigger set is NP-hard. The de Bruijn relationship provides efficiently computable structural bounds: the compacted de Bruijn graph is itself a feasible PFG and hence gives an upper bound on the optimum, while its segment-base sum gives a lower bound on the segment contribution of every PFG with the corresponding trigger length. When these bounds are close, a compacted de Bruijn graph is therefore a principled approximation to the best prefix-free representation at that scale.

The fixed-parameter algorithm in Observation~\ref{obs:fpt} suggests a complementary direction for exact methods: instead of parameterizing by the total pangenome length, one can parameterize by a measure of sequence diversity, such as the number of distinct candidate $k$-mers. Developing sharper FPT or approximation algorithms under pangenome-specific repetitiveness measures is a natural direction for future work.

\bibliographystyle{abbrv}
\bibliography{bibliography}

@article{compeau2015bioinformatics,
  title={Bioinformatics algorithms: an active learning approach},
  author={Compeau, Phillip and Pevzner, Pave},
  journal={(No Title)},
  year={2015}
}

@article{1983_theory_of_np,
    title = {Michael R. $\Pi$Garey and David S. Johnson. Computers and intractability. A guide to the theory of NP-completeness. WH Freeman and Company, San Francisco1979, x+ 338 pp.},
    author = {Lewis, Harry R},
    journal = {The Journal of Symbolic Logic},
    volume = {48},
    number = {2},
    pages = {498--500},
    year = {1983},
    publisher = {Cambridge University Press},
}

@article{2019bigBWT,
    title = {Prefix-free parsing for building big BWTs},
    author = {Boucher, Christina and Gagie, Travis and Kuhnle, Alan and Langmead, Ben and Manzini, Giovanni and Mun, Taher},
    journal = {Algorithms for Molecular Biology},
    volume = {14},
    pages = {1--15},
    year = {2019},
    publisher = {Springer},
}

@inproceedings{2023recursive_bigBWT,
    title = {Recursive prefix-free parsing for building big BWTs},
    author = {Oliva, Marco and Gagie, Travis and Boucher, Christina},
    booktitle = {2023 data compression conference (DCC)},
    pages = {62--70},
    year = {2023},
    organization = {IEEE},
}

@inproceedings{2023pfp_lz77,
    title = {LZ77 via prefix-free parsing},
    author = {Hong, Aaron and Rossi, Massimiliano and Boucher, Christina},
    booktitle = {2023 Proceedings of the Symposium on Algorithm Engineering and Experiments (ALENEX)},
    pages = {123--134},
    year = {2023},
    organization = {SIAM},
}

@article{2023pfp_fm_index,
    title = {Acceleration of FM-Index Queries Through Prefix-Free Parsing},
    author = {Hong, Aaron and Oliva, Marco and K{\"o}ppl, Dominik and Bannai, Hideo and Boucher, Christina and Gagie, Travis},
    journal = {arXiv preprint arXiv:2305.05893},
    year = {2023},
}

@article{2024recursive_pfp_pangenome,
    title = {Building a pangenome alignment index via recursive prefix-free parsing},
    author = {Ferro, Eddie and Oliva, Marco and Gagie, Travis and Boucher, Christina},
    journal = {iScience},
    volume = {27},
    number = {10},
    year = {2024},
    publisher = {Elsevier},
}

@inproceedings{2021pfp_cst,
    title = {PFP Compressed Suffix Trees},
    author = {Boucher, Christina and Cvacho, Ond{\v{r}}ej and Gagie, Travis and Holub, Jan and Manzini, Giovanni and Navarro, Gonzalo and Rossi, Massimiliano},
    booktitle = {2021 Proceedings of the Workshop on Algorithm Engineering and Experiments (ALENEX)},
    pages = {60--72},
    year = {2021},
    organization = {SIAM},
}

@inproceedings{2019rpair,
    title = {Rpair: rescaling RePair with rsync},
    author = {Gagie, Travis and I, Tomohiro and Manzini, Giovanni and Navarro, Gonzalo and Sakamoto, Hiroshi and Takabatake, Yoshimasa},
    booktitle = {International Symposium on String Processing and Information Retrieval},
    pages = {35--44},
    year = {2019},
    organization = {Springer},
}

@inproceedings{2022pfp_wg,
    title = {Prefix-Free Parsing for Building Large Tunnelled Wheeler Graphs},
    author = {Goga, Adri{\'a}n and Bal{\'a}{\v{z}}, Andrej},
    booktitle = {22nd International Workshop on Algorithms in Bioinformatics (WABI 2022)},
    year = {2022},
    organization = {Schloss Dagstuhl-Leibniz-Zentrum f{\"u}r Informatik},
}

@inproceedings{2023pfg,
    title = {Prefix-free graphs and suffix array construction in sublinear space},
    author = {Bal\'{a}\v{z}, Andrej and Petescia, Alessia},
    booktitle = {Workshop on Bioinformatics and Computational Biology WBCB 2023},
    year = {2023},
}

@article{2024recursive_repair,
    title = {Recursive RePair: Increasing the Scalability of RePair by Decreasing Memory Usage},
    author = {Kim, Justin and Varki, Rahul and Oliva, Marco and Boucher, Christina},
    journal = {bioRxiv},
    pages = {2024--07},
    year = {2024},
    publisher = {Cold Spring Harbor Laboratory},
}

@ARTICLE {2021repmeasure,
        TITLE = "Indexing Highly Repetitive String Collections, 
		 Part {I}: Repetitiveness Measures",
        AUTHOR = "G. Navarro",
        JOURNAL = {ACM Computing Surveys},
        YEAR = 2021,
        VOLUME = 54,
        NUMBER = 2,
        PAGES = "article 29"
}

@ARTICLE {2021compress,
        TITLE = "Indexing Highly Repetitive String Collections, 
		 Part {II}: Compressed Indexes",
        AUTHOR = "G. Navarro",
        JOURNAL = {ACM Computing Surveys},
        YEAR = 2021,
	VOLUME = 54,
	NUMBER = 2,
	PAGES = "article 26"
}

@article{luca2025measuring,
  title={Measuring genomic data with prefix-free parsing},
  author={Luc{\`a}, Simone and Masillo, Francesco and Lipt{\'a}k, Zsuzsanna},
  journal={Computational Biology and Chemistry},
  pages={108870},
  year={2025},
  publisher={Elsevier}
}

@article{computational2018computational,
  title={Computational pan-genomics: status, promises and challenges},
  journal={Briefings in bioinformatics},
  volume={19},
  number={1},
  pages={118--135},
  year={2018},
  publisher={Oxford University Press}
}

@article{tettelin2005pangenome,
  title={Genome analysis of multiple pathogenic isolates of Streptococcus agalactiae: implications for the microbial “pan-genome”},
  author={Tettelin, Herv{\'e} and Masignani, Vega and Cieslewicz, Michael J and Donati, Claudio and Medini, Duccio and Ward, Naomi L and Angiuoli, Samuel V and Crabtree, Jonathan and Jones, Amanda L and Durkin, A Scott and others},
  journal={Proceedings of the National Academy of Sciences},
  volume={102},
  number={39},
  pages={13950--13955},
  year={2005},
  publisher={National Academy of Sciences}
}

@book{garrison2018graphical,
  title={Graphical pangenomics},
  author={Garrison, Erik Peter},
  year={2018},
  publisher={University of Cambridge (United Kingdom)}
}

@article{equi2023graphs,
  title={Graphs cannot be indexed in polynomial time for sub-quadratic time string matching, unless SETH fails},
  author={Equi, Massimo and M{\"a}kinen, Veli and Tomescu, Alexandru I},
  journal={Theoretical Computer Science},
  volume={975},
  pages={114128},
  year={2023},
  publisher={Elsevier}
}

@inproceedings{iliopoulos2017efficient,
  title         = {Efficient pattern matching in elastic-degenerate texts},
  author        = {Iliopoulos, Costas S and Kundu, Ritu and Pissis, Solon P},
  booktitle     = {International Conference on Language and Automata Theory and Applications},
  pages         = {131--142},
  year          = {2017},
  organization  = {Springer},
}

@inproceedings{bernardini2017pattern,
  title={Pattern matching on elastic-degenerate text with errors},
  author={Bernardini, Giulia and Pisanti, Nadia and Pissis, Solon P and Rosone, Giovanna},
  booktitle={International Symposium on String Processing and Information Retrieval},
  pages={74--90},
  year={2017},
  organization={Springer}
}

@article{makinen2020linear,
  title         = {Linear time construction of indexable founder block graphs},
  author        = {M{\"a}kinen, Veli and Cazaux, Bastien and Equi, Massimo and Norri, Tuukka and Tomescu, Alexandru I},
  journal       = {arXiv preprint arXiv:2005.09342},
  year          = {2020},
}

@article{2011paten1,
  title         = {Cactus graphs for genome comparisons},
  author        = {Paten, Benedict and Diekhans, Mark and Earl, Dent and John, John St and Ma, Jian and Suh, Bernard and Haussler, David},
  journal       = {Journal of Computational Biology},
  volume        = {18},
  number        = {3},
  pages         = {469--481},
  year          = {2011},
  publisher     = {Mary Ann Liebert, Inc. 140 Huguenot Street, 3rd Floor New Rochelle, NY 10801 USA},
}

@article{2011paten2,
  title         = {Cactus: Algorithms for genome multiple sequence alignment},
  author        = {Paten, Benedict and Earl, Dent and Nguyen, Ngan and Diekhans, Mark and Zerbino, Daniel and Haussler, David},
  journal       = {Genome research},
  volume        = {21},
  number        = {9},
  pages         = {1512--1528},
  year          = {2011},
  publisher     = {Cold Spring Harbor Lab},
}

@article{gagie2017wheeler,
  title={Wheeler graphs: A framework for BWT-based data structures},
  author={Gagie, Travis and Manzini, Giovanni and Sir{\'e}n, Jouni},
  journal={Theoretical computer science},
  volume={698},
  pages={67--78},
  year={2017},
  publisher={Elsevier}
}

@inproceedings{cotumaccio2021indexing,
  title={On indexing and compressing finite automata},
  author={Cotumaccio, Nicola and Prezza, Nicola},
  booktitle={Proceedings of the 2021 ACM-SIAM Symposium on Discrete Algorithms (SODA)},
  pages={2585--2599},
  year={2021},
  organization={SIAM}
}

@article{kornblum2006identifying,
  title         = {Identifying almost identical files using context triggered piecewise hashing},
  author        = {Kornblum, Jesse},
  journal       = {Digital investigation},
  volume        = {3},
  pages         = {91--97},
  year          = {2006},
  publisher     = {Elsevier},
}

\end{document}